\documentclass[twoside,usletter,11pt,leqno,onecolumn]{article}  
\usepackage{ltexpprt} 
\usepackage[margin=1in]{geometry}
\usepackage{enumerate}
\usepackage{tikz}
\usepackage{url}
\usepackage{rotating}
\usepackage{wrapfig}
\usepackage{amsmath}
\usetikzlibrary{positioning,shapes.geometric}

\newcommand{\topklists}[0]{top $k$ lists}
\newcommand{\ie}[0]{\textit{i.e.,}}
\newcommand{\eg}[0]{\textit{e.g.,}}
\newcommand{\ignore}[1]{}
\newcommand{\topklist}[0]{top $k$ list}

\begin{document}

\setcounter{footnote}{1}
\title{\Large An information measure for comparing top $k$ lists}
\author{Arun S. Konagurthu~\footnote{Clayton School of Computer Science and Information Technology, Monash University, VIC 3800 Australia.}
\setcounter{footnote}{0}
\footnote{Correspondence: arun.konagurthu@monash.edu}\\
\and 
James H. Collier $^\dagger$}
\date{}

\maketitle


\begin{abstract} 
Comparing the top $k$ elements between two or more ranked results is a common task in many contexts and settings.  A few measures have been proposed to compare top $k$ lists with attractive mathematical properties, but they face a number of pitfalls and shortcomings in practice. This work introduces a new measure to compare any two top $k$ lists based on measuring the information these lists convey. Our method investigates the compressibility of the lists, and the length of the message to losslessly encode them gives a natural and robust measure of their variability.  This information-theoretic measure objectively reconciles all the main considerations that arise when measuring  (dis-)similarity between lists: the extent of their non-overlapping elements in each of the lists; the amount of disarray among overlapping elements between the lists;  the measurement of displacement of actual ranks of their overlapping elements. 
\end{abstract}

\section{Introduction}
\subsection{Motivation}
Ranked results are handled in diverse settings, from the web page results of search engines to genes in differential gene co-expression experiments. A routine task that emerges from such ranking results is the assessment of variability among its top few ($k$) elements between two of more such rankings. 

This topic has received much attention over the past decade, mainly in the context of information retrieval. Among the most cited work on this topic is that of Fagin and colleagues~\cite{fagin2003comparing}. They proposed an easy to compute metric based primarily on Spearman's foot rule~\cite{spearman1904proof}. Formally,  if $\pi_1$ and $\pi_2$ define two permutations from the symmetric group $S_n$ of all permutations of $n$ elements, Spearman's foot rule gives the $L_1$ distance between the ranks of corresponding elements in the two permutations, as:
$
L_1(\pi_1,\pi_2) = \sum_{i=1}^{n} |\pi_1(i)-\pi_2(i)|,
$
where any $\pi_1(i)$ or $\pi_2(i)$ is the position (rank) of the $i$th element in the permutation, given some total ordering of $n$ elements. Fagin \textit{et al.} extend this metric to comparing two \topklists\ in presence of non-overlapping elements (\ie\ elements that are in one list but not in the other). This is achieved by fixing the contribution, to the distance,  of the non-overlapping elements to a value greater than $k$, typically $(k+1)$. Formally, the extended metric for two \topklists\ $\tau_1$ and $\tau_2$ is defined as
$$
L_1(\tau_1,\tau_2) = 2(k-|\tau_1\cap\tau_2|)(k+1)+\sum_{i\in \tau_1\cap\tau_2} |\tau_1(i)-\tau_2(i)| - \sum_{i\in \tau_1-\tau_2} \tau_1(i) - \sum_{i\in\tau_2-\tau_1} \tau_2(i)
$$
where $\tau_1\cap\tau_2$ is the set of elements that overlap between the two lists, $|\tau_1\cap\tau_2|$ denotes the number of overlapping elements, $\tau_1-\tau_2$ gives the non-overlapping elements in $\tau_1$, and $\tau_2-\tau_1$ gives those in $\tau_2$.

Although this measure can  be shown to have good mathematical properties , in practice, it has crucial limitations. Mainly, it can be seen from the formulation that the term $2(k-|\tau_1\cap\tau_2|)(k+1)$ grows quadratically for increasing values of $k$ and decreasing proportion of overlapping elements. In fact, in many applications requiring comparison of \topklists\ (\eg\ web search results), non-overlapping elements form a significant proportion of the lists. Furthermore, this metric is insensitive to the absolute ranks of the overlapping elements in the respective lists; when computing the $L_1$ distance, the overlapping elements are re-ranked, and hence ignore the displacement of these elements when comparing two lists.

Other mathematically attractive metrics have also been proposed; for example, those based on Kendall tau distance~\cite{kendall1938new}.  Colloquially, this distance is called the \textit{bubble-sort distance} since it measures the number of adjacent transpositions required to convert (i.e., sort) one permutation to another.  Formally, for any two permutations $\pi_1$ and $\pi_2$, Kendall tau distance is defined (using the same notations as above) as
$
K(\pi_1,\pi_2) = \sum_{\forall 1\le i<j\le n} \kappa_{i,j}(\pi_1,\pi_2),
$
where $\kappa_{i,j}(\pi_1,\pi_2) = 0$ if $\pi_1(i)<\pi_1(j)$ and $\pi_2(i) <\pi_2(j)$, or $\kappa_{i,j}(\pi_1,\pi_2) = 1$ otherwise. Extending this idea, the following cost function was proposed to compare two \topklists~\cite{fagin2006comparing}:
$$
K(\tau_1,\tau_2) = (k-|\tau_1\cap\tau_2|)((2+p)k-p|\tau_1\cap\tau_2|+1-p)+\sum_{i\in \tau_1\cap\tau_2} \kappa_{i,j}(\tau_1,\tau_2) - \sum_{i\in \tau_1-\tau_2} \tau_1(i) - \sum_{i\in\tau_2-\tau_1} \tau_2(i)
$$
where, $p$ is a tunable penalty parameter to account for the transposition distance between non-overlapping elements in $\tau_1$ and $\tau_2$. However, it is easy to see that this metric is also sensitive to the size of non-overlapping elements in the two lists, in addition to the choice of penalty parameter $p$.

There is further work on this problem, mainly designing measures targeted at focussed applications~\cite{bar2006methods,budinskapackage,fury2006overlapping,pearson2007reciprocal,jurman2009canberra,jurman2012algebraic}. Noteworthy among these is the use of Canberra distance~\cite{lance1966computer} to measure distance between \topklists~\cite{jurman2012algebraic}. This distance is a weighted variant of Spearman's $L_1$ distance, which ensures that the displacement of elements with higher ranks results in a greater penalty compared to those with lower ranks. 

\subsection{Our results}
In this paper, we introduce a new information measure to compare any two top $k$ lists.
We build our method on the statistical framework of minimum length encoding introduced by Chris Wallace~\cite{wallace68,WallaceBook}. Our method investigates the compressibility of \topklists. It is intuitive to see that closely related lists have more information in common (and hence more compressible) than the lists that are poorly related. Thus, the length of the lossless encoding message gives a natural and rigorous measure to estimate the variability between two lists. 
Unlike previous work, this measure implicitly allows an \textit{objective} trade-off between conflicting criteria when measuring the variability between two lists. Mainly, these include: (1) the measurement of the extent of non-overlap in the two lists, (2) the measurement of disarray of its overlapping elements, and (3) the displacement of the positions (ranks) of these elements.

We note that measuring the \textit{true} information content of any data is incomputable. This follows from the fact that Solomonoff-Kolmogorov-Chaitin Complexity~\cite{kolmogorov1963tables,solomonoff1964formal,chaitin1966length} is undecidable. However, effective and efficient statistical models for data compression provide reasonable upper bounds (\ie\ estimates) of true information content.

Our paper provides an approach to estimating the information content in any given pair of \topklists. To keep this approach general, our models of compression use \textit{bland} assumptions and priors. However it is important to note that this information theoretic framework  can be adapted to individual contexts by accommodating prior knowledge about rankings in those settings.

\subsection{Organization of the paper}
This paper is organized as follows.
Section \ref{sec:information} introduces our information measure formally and describes some interesting mathematical properties. Section \ref{sec:practical} explains the practical details involved in estimating the information content of two lists. Section \ref{sec:results} presents the results of comparing this measure with other popular distance metrics on ranked lists.

\section{Information measure on comparing ranked lists}
\label{sec:information}

\begin{Definition} (Information content of an outcome)\\ 
\label{obs:neglogprob}
Information conveyed by any outcome or event $E$ whose probablity is $P(E)$  is given by $I(E) = -\log\left(P(E)\right)$.\footnote{Base of the logarithm gives the information measure its units. $\log_2$ yields information measured in bits. $\ln$ gives nits or nats, and  $\log_{10}$, dits or hartleys.}
\end{Definition}
This is a fundamental result of Shannon's seminal work on the theory of communication~\cite{shannon48}. 

\begin{lemma} (Measure of Information between two top $k$ lists)\\
For two \topklists, $\tau_1$ and $\tau_2$, the total amount of information contained in them is  
$
I(\tau_1, \tau_2) = 
  \min\left\{
    I(\tau_1)+ I(\tau_2),
    I(\tau_1) + I(\tau_2|\tau_1), 
    \right\}
$
\end{lemma} 
\begin{proof}
\noindent Two scenarios arise in measuring the joint information in the \topklists. The lists are either independent of each other, or they  are related.

\begin{enumerate}[(i)]
\item If $\tau_1$ and $\tau_2$ are independent of each other, that is the knowledge of one list does not inform the contents of the other list, the joint information content in these lists is the sum of the information content in each of the lists taken separately, \ie\ $I(\tau_1)+I(\tau_2)$. Lets term this expression as the null model message length, denoted by $NULL(\tau_1,\tau_2)$. 

\item On the other hand, if $\tau_1$ and $\tau_2$ are related to each other, the knowledge of one list informs the contents of the other, to a less or more extent. Hence the additional information required to convey $\tau_2$ will be \textit{less} than stating it independently. From Bayes's theorem we have $P(\tau_1,\tau_2) = P(\tau_1)P(\tau_2|\tau_1)$. Applying Observation 1.1, we get  $I(\tau_1, \tau_2) = I(\tau_1) + I(\tau_2|\tau_1)$.
\end{enumerate}

It follows that the amount of information in the two lists is bounded above by $NULL(\tau_1,\tau_2)$ and is measured as $I(\tau_1) + I(\tau_2|\tau_1)$. 
In addition, the following also holds:
$I(\tau_1,\tau_1) \le I(\tau_1,\tau_2) \le I(\tau_1)+I(\tau_2)$
\end{proof}

\begin{Definition} (Information cost)\\
We measure information cost (or divergence) between two \topklists\ as $I(\tau_1,\tau_2)-I(\tau_1,\tau_1) = I(\tau2|\tau_1) - I(\tau_1|\tau_1)$.
\end{Definition}

\begin{property}
For three \topklists, $\tau_1$, $\tau_2$ and $\tau_3$,
$I(\tau_1,\tau_2) - I(\tau_1,\tau_3) = \log\left(\frac{P(\tau_3|\tau_1)}{P(\tau_2|\tau_1)}\right)$
\end{property}
From Lemma 2.1, we have
$I(\tau_1,\tau_2) = I(\tau_1) + I(\tau_2|\tau_1)$ and 
$I(\tau_1,\tau_3) = I(\tau_1) + I(\tau_3|\tau_1)$

Subtracting the two,
\begin{eqnarray*}
I(\tau_1\tau_2) - I(\tau_1,\tau_3) &=& I(\tau_2|\tau_1) - I(\tau_3|\tau_1)\\
&=& -\log\left(P(\tau_2|\tau_1)\right) -\log\left(P(\tau_3|\tau_1)\right)\\
&=& \log\left(\frac{P(\tau_3|\tau_1)}{P(\tau_2|\tau_1)}\right)
\end{eqnarray*}
This property gives good foundation to compare any two lists.

\begin{property} (Measure of information is symmetric)\\
For optimal encodings of $\tau_1$, $\tau_2$, $\tau_1|\tau_2$ and $\tau_2|\tau_1$, $I(\tau_1,\tau_2)\equiv I(\tau_2,\tau_1)$. 
\end{property}
This follows from Bayes's theorem. $P(\tau_1,\tau_2) = P(\tau_1)P(\tau_2|\tau_1) = P(\tau_2)P(\tau_1|\tau_2)$. Applying Definition 2.1 to Bayes's, this property follows. However, as discussed earlier, this property holds only when dealing with  true measures of information content.  From the wide data compression literature, it can be seen that symmetry holds approximately up to some constant, which is dependent on the encoding scheme rather than the data itself.

A corollary of this property is that the \textit{conditional information} between $\tau_1$ and $\tau_2$ is \textit{not} symmetric: $I(\tau_1|\tau_2) = I(\tau_2|\tau_1) + \delta$, where $\delta = I(\tau_2)-I(\tau_1)$ is dependent on the information content in the respective lists. 

\begin{property} (Directed acyclic triangular inequality of conditional information)\\
For three \topklists, $\tau_1$,  $\tau_2$, and $\tau_3$, we have:\\
\begin{center}
$
I(\tau_1|\tau_2) \le  
I(\tau_1|\tau_3) +  
I(\tau_3|\tau_2)   
$
{
\begin{tikzpicture}[scale=6]
\ignore{[%
->,
shorten >=2pt,
>=stealth,
node distance=1cm,
pil/.style={
->,
thick,
shorten =2pt,}
]}
\node (2) {$\tau_2$};
\node[left =of 2] (1) {$\tau_1$};
\node[right =of 2] (3) {$\tau_3$};
\draw[<-] (1.east) --(2.west);
\draw[->] (2.east) --(3.west);
\draw[<-] (1) to [out=15,in=165] (3);
\node[above right] at (-0.1,0.05){\tiny$~I(\tau_1|\tau_3)$};
\node[below left] at (-0.025,-0.015){\tiny$~I(\tau_1|\tau_2)$};
\node[below right] at (+0.025,-0.015){\tiny$~I(\tau_3|\tau_2)$};
\end{tikzpicture}
}
\end{center}
\end{property} 
This follows by expanding the joint information in the three lists as follow:
\begin{eqnarray*}
I(\tau_1,\tau_2,\tau_3) 
&=& I(\tau_3) + I(\tau_1,\tau_2|\tau_3)
= I(\tau_2) + I(\tau_1,\tau_3|\tau_2)\\
~&=& I(\tau_3) + I(\tau_1|\tau_3) + I(\tau_2|\tau_1,\tau_3)
= I(\tau_2) + I(\tau_1,\tau_3|\tau_2)\\
~&=& I(\tau_3) + I(\tau_1|\tau_3) + I(\tau_2|\tau_3) 
\ge I(\tau_2) + I(\tau_1,\tau_3|\tau_2)\\
\end{eqnarray*}
Rearranging terms, we get:
\begin{eqnarray*}
\left(I(\tau_3) + I(\tau_2|\tau_3)\right) &+& I(\tau_1|\tau_3) 
\ge I(\tau_2) + I(\tau_1,\tau_3|\tau_2)\\
I(\tau_2,\tau_3) &+& I(\tau_1|\tau_3) 
\ge I(\tau_2) + I(\tau_1,\tau_3|\tau_2)\\
\left(I(\tau_2,\tau_3) - I(\tau_2)\right) &+& I(\tau_1|\tau_3) 
\ge  I(\tau_1,\tau_3|\tau_2)\\
I(\tau_3|\tau_2) &+& I(\tau_1|\tau_3) 
\ge  I(\tau_1,\tau_3|\tau_2)\\
I(\tau_3|\tau_2) &+& I(\tau_1|\tau_3) 
\ge  I(\tau_1|\tau_2) + I(\tau_3|\tau_1,\tau_2)\\
I(\tau_3|\tau_2) &+& I(\tau_1|\tau_3) 
\ge  I(\tau_1|\tau_2) \\
\end{eqnarray*}

\begin{property} (Near-coincidence of conditional information)\\
For any given \topklist\ $\tau$, $I(\tau|\tau) = \epsilon$, where $\epsilon$ is some small constant which is independent of the information in $\tau$.
\end{property}
This follows because $I(\tau,\tau)=I(\tau)+I(\tau|\tau)$. Knowing the \topklist\ $\tau$, the additional conditional information required to state a copy of itself is a very small constant.

\section{Practical Considerations}
\label{sec:practical}
This section will describe an approach to realize an information theoretic measure to quantify the variability of two \topklists. To make this measure intuitively understood, we describe the details as a communication process between an imaginary pair of transmitter and receiver connected over a Shannon channel.

Alice has access to two \topklists\ $\tau_1$ and $\tau_2$. Alice's goal is to communicate the information in both these lists to Bob \textit{exactly} as she see it. To achieve this, Alice aims to constructs a two-part message. In the first, she will transmit $\tau_1$ taking $I(\tau_1)$ bits. In the second, she aims to exploit the redundancy (if any) between the lists so that $\tau_2$ can be transmitted more concisely; this takes $I(\tau_2|\tau_1)$ bits. 

In this information theoretic framework the measure of (dis-)similarity between two \topklists\ is the total length of this two-part message, $I(\tau_1)+I(\tau_2|\tau_1)$. It is easy to see that if $\tau_2=\tau_1$, the second part is extremely concise. On the other hand, if $\tau_2$ is completely unrelated to $\tau_1$, then $I(\tau_2|\tau_1)$ cannot be better (\ie\ shorter) than  $I(\tau_2)$.

For Alice to transmit the two lists, $\tau_1$ and $\tau_2$ losslessly, the following information needs to be transmitted:
\begin{enumerate}
 \setlength{\itemsep}{1pt}
  \setlength{\parskip}{0pt}
  \setlength{\parsep}{0pt}
\item The size $k = |\tau_1| = |\tau_2|$ of the lists.
\item The elements in  $\tau_1$, in the order they appear.
\item The overlapping elements between $\tau_1$ and $\tau_2$.
\item The absolute positions of these overlapping elements in $\tau_2$.
\item The permutation of overlapping elements in $\tau_2$ with respect to the order defined by $\tau_1$
\item The non-overlapping elements in $\tau_2$ in the order they appear.
\end{enumerate}

Two distinct cases have to be handled to formulate encoding schemes for each of the above. (1) When the domain of elements that are being ranked, of which $\tau_1$ and $\tau_2$ are (partial) instances, is \textit{known}. For instance, consider the rankings of top 50 differentially expressed genes in a differential co-expression experiment. Here the total domain of gene and their labels (identifiers)
is known. (2) Conversely, when the domain of ranked elements remains \textit{unknown}. For instance, consider the search results from popular web search engines. While we see the top search results, the number of pages each search engine indexes is variable and could be fewer than the pages available on the internet.  

The remaining part of this section, we handle these two cases and describe encoding schemes to transmit for each of the enumerated pieces of information.

\subsection{Case 1: When the domain of elements is known.}
Here we assume that the size ($N$) of the domain is known along with the labels (or identifiers) of elements in it.

~\\
\noindent\textbf{\textit{Step 1: Transmitting the size of the \topklists.}} The size of $k\le N$ is transmitted as an integer code. Since both Alice and Bob know that the \topklists\ come from a domain of $N$ elements, a simple encoding of $k$  takes $\log(N)$ bits, assuming an uniform distribution over the choices of $k$ in the range $1\le k\le N$. We note that more sophisticated encodings can be conceived if their is a prior belief that the distribution of $k$ is non-uniform.

~\\
\noindent\textbf{\textit{Step 2: Transmitting $\tau_1$.}}
Then, transmitting the information in $\tau_1$ can be achieved by communicating, over an integer code, the lexicographic number associated with $\tau_1$ in some (mutually agreed) lexicographic ordering of the $k$-permutations of $N$ elements. Since both Alice and Bob know the domain from which the ranking was generated,  the lexicographic ordering  of $k$-permutations can be treated as a part of the code book of communication, and need not be transmitted. 

~\\
\noindent\textbf{\textit{Step 3: Transmitting overlapping elements between $\tau_1$ and $\tau_2$.}}
At this stage Bob already knows $\tau_1$. To nominate the overlapping elements, that is, the intersection between the two \topklists, a bit mask \texttt{$b_1$} is defined where the set bits indicate the positions in $\tau_1$ where the overlapping elements reside. Transmission complexity of stating the intersection between $\tau_1$ and $\tau_2$  is same as the complexity of this bit mask. An efficient encoding scheme to transmit this bit mask, assuming no prior knowledge about the distribution of the set bits, would be using an adaptive code over a binomial distribution. 

The mask $b_1$ is a binary sequence of length $k$. The adaptive encoding requires maintaining two running counters that count incrementally the number of 0s and number of 1s, starting from an initial value of 1. Traversing the bit mask left to right, for every symbol in $b_1$, Alice estimates its probability by dividing the current state of the symbol's counter by the sum of the two counters.  After the probability is estimated, Alice increments the corresponding counter by 1. The code length to state each symbol is the negative logarithm of its estimated probability. Generalizing this, if $cnt[0]$ is the number of 0s and $cnt[1]$ be the number of 1s in any bit mask of size $k$, then the length of the message to transmit this bit mask is 
$-\log_2\left(
\frac{cnt[0]!\times cnt[1]!}{(k+1)!}
\right)$ bits. Figure \ref{fig:adaptivecode}(a) gives an example.
Notice that both Alice and Bob initialize their counters to 1. Alice encodes each symbol in the bit mask and transmits it before incrementing the corresponding counter at her end. Bob decodes the received symbol using the same estimate of the probability and updates the counters on his side, thus keeping both counters synchronized to achieve a lossless communication.

\begin{figure}
\centering
\begin{tabular}{cc}
\begin{tabular}{|l|@{~}c@{~}c@{~}c@{~}c@{~}c@{~}c@{~}c@{~}c@{~}c@{~}c|}
\hline
 $b_1$      & 0           & 0           & 1           & 1           & 0           & 0           & 1           & 0           & 0            & 0\\\hline
 $cnt[0]$   & 1           & 2           & 3           & 3           & 3           & 4           & 5           & 5           & 6            & 7\\
 $cnt[1]$   & 1           & 1           & 1           & 2           & 3           & 3           & 3           & 4           & 4            & 4\\\hline
 Prob.   &$\frac{1}{2}$&$\frac{2}{3}$&$\frac{1}{4}$&$\frac{2}{5}$&$\frac{3}{6}$&$\frac{4}{7}$&$\frac{3}{8}$&$\frac{5}{9}$&$\frac{6}{10}$&$\frac{7}{11}$\\
\hline
\end{tabular}
&
\begin{tabular}{|l|@{~}c@{~}c@{~}c@{~}c@{~}c@{~}c@{~}c@{~}c@{~}c@{~}c|}
\hline
 $b_1$                      & 0           & 0           & 1           & 1           & 0           & 0           & 1            & 0            & 0            & 0\\
 $b_2$                      & 0           & 0           & 1           & 1           & 1           & 0           & 0            & 0            & 0            & 0\\\hline
 $cnt[0|0]$ or $cnt[1|1]$   & 1           & 2           & 3           & 4           & 5           & 5           & 6            & 6            & 7            & 8\\
 $cnt[0|1]$ or $cnt[1|0]$   & 1           & 1           & 1           & 1           & 1           & 2           & 2            & 3            & 3            & 3\\\hline
 Prob.                      &$\frac{1}{2}$&$\frac{2}{3}$&$\frac{3}{4}$&$\frac{4}{5}$&$\frac{1}{6}$&$\frac{5}{7}$&$\frac{2}{8}$&$\frac{6}{9}$&$\frac{7}{10}$&$\frac{8}{11}$ \\
\hline
\end{tabular} \\
(a)&(b) \\
\end{tabular}
\caption{Examples of the adaptive encoding schemes for bit masks described in the main text.}
\label{fig:adaptivecode}
\end{figure}

~\\
\noindent\textbf{\textit{Step 4: Transmitting absolute positions in $\tau_2$ of the overlapping elements.}}
This again defines another bit mask, $b_2$. It is easy to see that there are $\binom{k}{cnt[1]}$ possible candidates for $b_2$, given that Bob already knows $b_1$.  Therefore, assuming these candidates are uniformly distributed, the optimal message length to state $b_2$ takes $\log\binom{k}{cnt[1]}$ bits.
We emphasize here that $b_2$ ignores the permutation of the overlapping elements as they appear in $\tau_2$ (with respect to $\tau_1$) -- this is handled in the next step.

While the above encoding is optimal, it, however, does not account for the displacement of overlapping elements in terms of their absolute ranks in the list. It might arise in some applications that the displacement is among the criteria of comparing two lists. Hence we propose a modified adaptive scheme to account for this displacement. We use two  counters; the first tracks the number of times the symbols in bit masks $b_1$ and $b_2$ remain the same at a given position (column); the second tracks the number of times they are different. These counters are used to estimate the probabilities while traversing along $b_2$. See Figure \ref{fig:adaptivecode}(b) for an example.

\begin{figure}[h!]
\centering
\begin{tabular}{||r|rccl||r|rccl||r|rccl||r|rccl||}
\hline\hline
 $0_{10}$&a&b&c&d  & $6_{10}$&b&a&c&d & $12_{10}$&c&a&b&d & $18_{10}$&d&a&b&c\\
   &(0&0&0&0)$_!$  &       &(1&0&0&0)$_!$ &  &(2&0&0&0)$_!$&   &(3&0&0&0)$_!$\\
         \hline
 $1_{10}$&a&b&d&c & $7_{10}$&b&a&d&c & $13_{10}$&c&a&d&b & $19_{10}$&d&a&c&b\\
         &(0&0&1&0)$_!$ &         &(1&0&1&0)$_!$ &          &(2&0&1&0)$_!$ &          &(3&0&1&0)$_!$\\
         \hline
 $2_{10}$&a&c&b&d & $8_{10}$&b&c&a&d & $14_{10}$&c&b&a&d & $20_{10}$&d&b&a&c\\
         &(0&1&0&0)$_!$ &         &(1&1&0&0)$_!$ &          &(2&1&0&0)$_!$ &          &(3&1&0&0)$_!$\\
         \hline
 $3_{10}$&a&c&d&b & $9_{10}$&b&c&d&a & $15_{10}$&c&b&d&a & $21_{10}$&d&b&c&a\\
         &(0&1&1&0)$_!$ &         &(1&1&1&0)$_!$ &          &(2&1&1&0)$_!$ &          &(3&1&1&0)$_!$\\
         \hline
 $4_{10}$&a&d&b&c &$10_{10}$&b&d&a&c & $16_{10}$&c&d&a&b & $22_{10}$&d&c&a&b\\
         &(0&2&0&0)$_!$ &         &(1&2&0&0)$_!$ &          &(2&2&0&0)$_!$ &          &(3&2&0&0)$_!$\\
         \hline
 $5_{10}$&a&d&c&b &$11_{10}$&b&d&c&a & $17_{10}$&c&d&b&a & $23_{10}$&d&c&b&a\\
         &(0&2&1&0)$_!$ &         &(1&2&1&0)$_!$ &          &(2&2&1&0)$_!$ &          &(3&2&1&0)$_!$\\
         \hline\hline
\end{tabular}
\caption{All possible permutation of elements \texttt{a,b,c,d}, their lexicographical number in base 10, and their corresponding sequence of digits in a factorial number system. The factoradic system defines a bijection between the permutation and its lexicographic number. For example, $21_{10} = (3,1,1,0)_! =  3\times 3! + 1\times 2! + 1\times 1!+ 0.$}
\label{fig:factoradic}
\end{figure}

~\\
\noindent\textbf{\textit{Step 5: Transmitting the permutation of overlapping elements in $\tau_2$ with respect to $\tau_2$.}}
From the previous step, Bob knows what the overlapping elements between the lists are, but does not know in what order they appear in $\tau_2$. To transmit the permutation of these overlapping elements efficiently, a lexicographic numbering can be mutually agreed between them (as a part of the code book). Then, transmitting the permutation of these elements requires simply communicating its lexicographic number over some integer code. 

However, to make this transmission efficient 
a factoradic (or mixed factorial base numbering system) can be employed~\cite{donald1999art}. This system defines a bijection between the symmetric group $S_n$ to $n!$ possible permutations in that group.  

Concretely, let $\pi= \{\pi(i),\pi(2),\cdots,\pi(n)\}$ be some permutation of $n$ symbols, where $\pi(i)$ is the rank of the $i$th element in the permutation. A factoradic of $\pi$ defines a sequence $f(\pi) = (f^1, f^2, \cdots, f^n)_!$, where any  $f^i$ is the number of $j$s  greater than $i$ such that $\pi(i)<\pi(j)$. See Figure \ref{fig:factoradic} for an example of a lexicographic ordering of the symmetric group $S_4$ labeled by elements `\texttt{a,b,c,d}', along with its corresponding factoradic sequence of digits. It can be observed that each factoradic digit $f^i$  denotes the number of successive \textit{adjacent transpositions}  on $\pi$ required to move each $\pi(i)$th element into its correct position.  The permutation index (in decimal) can be computed from a factoradic as $\sum_{i=0}^n f^i\times(n-i-1)!$. 

The sequence of digits $f(\pi)$ has several interesting properties. It has been shown that, if permutations in $S_n$ are distributed uniformly, each factoradic digit $f^i$ is also uniformly distributed in the range $0\le f^i\le (n-i-1)$~\cite{lehmer1960teaching}. Also, the factoradic digits $f^i$ are mutually independent of each other because they form projections on independent factors in the product $n\times(n-1)\times\cdots 1 \equiv n!$

Thus, transmitting a permutation of overlapping elements in $\tau_2$ involves transmitting its factoradic digits in sequence. For each factoradic digit $f^i$ in the range $0\le i< n$ (note: $f^n$ is always 0), any decreasing probability distribution on integers in that range can be used. Specifically,
we use a Wallace tree code~\cite{wallace1993coding} that defines a code over positive integers\footnote{Since factoradic digits start from 0, we just add 1 to each digit to map it to the Wallace tree code.} by associating each integer with the binary code used to uniquely identify a binary tree. (See \url{http://www.allisons.org/ll/MML/Discrete/Integers/} for more details.) Since this integer code is defined over the infinite space of positive integers, the probability associated with each code is normalized such that the total probability in the finite range $0\le f^i\le n-i-1$ adds up to 1.

~\\
\noindent\textbf{\textit{Step 6: Transmitting non-overlapping elements in $\tau_2$.}}
Given that the domain of elements is known and is of size $N$, each non-overlapping element can be stated in $\log_2(N-|\tau_1\cup\tau_2|)$ bits.  With this the communication process concludes.

\subsection{Case 2: When the domain of elements is unknown.}
\label{sec:case2}
Here Alice and Bob do not know the domain of elements being sorted. In the previous case, Steps 1,2 and 6 depended on knowing the domain, and hence require modification. The encodings for Steps 3,4, and 5 remain exactly the same as previously described.

Since this framework relies on lossless transmission, and there is no prior knowledge of the domain of possible labels in each of the two \topklists, this requires the lists (along with its labels) to be explicitly communicated.

To efficiently communicate $\tau_1$ and the non-overlapping elements in $\tau_2$, consider the  union of the two lists, $\tau_1\cup\tau_2$, such that the top $k$ elements define labels in $\tau_1$ (in that order) and the remainder are the labels of non-overlapping elements in the order they appear in $\tau_2$.

First, the size of the union $|\tau_1\cup\tau_2|$ is transmitted using the Wallace tree code defined over all positive integers. (This modifies previous Step 1.)
Then the labels in the $\tau_1\cup\tau_2$ can be compressed using, for instance, a standard, dictionary-based lossless data compression algorithm  of Lempel-Ziv-Welch ($LZW$)~\cite{ziv1978compression}.\footnote{A further entropy encoding  using Huffman's coding can be applied to the output symbols from $LZW$ compression.} The length, $|LZW(\tau_1\cup\tau_2)|$, in bits gives the cost to state the information in $\tau_1$ and non-overlapping elements in $\tau_2$. (This modifies previous Steps 2 and 6). 

\subsection{Time complexity}
In case 1: Steps 1,2, and 6 take $O(1)$ time to compute. The adaptive codes in Steps 3 and 4 take $O(k)$. In Step 5, finding the factoradic of a permutation of $n$ elements can be achieved in $O(n)$ time. (Refer \cite{myrvold2001ranking}; also note, when comparing top $k$ lists, $n\le k$.) Computing the code length of each factoradic takes $O(1)$ time.
Thus, the total time complexity to estimate the information content in the two lists grows as $O(k)$.

In case 2: Step 1 requires $O(1)$ time. Steps 2 and 6 are dealt together  and involves compression of labels in the set $\{\tau_1\cup\tau_2\}$. It can easily be seen that $k\le |\tau_1\cup\tau_2|\le 2k$. $LZW$ compression implemented naively has a time complexity of $O(SD)$, where $S$ is the number of input symbols and $D$ is the size of the dictionary. For most practical applications, $S$ is $O(k)$ in size and $D$ is a constant. Steps 3, 4, and 5 don't change from case 1 so will have the same time complexity. Thus, the total time complexity to estimate of information content is dominated by the $LZW$ compression which in practice grows as $O(k)$. 

\begin{figure}[t]
\begin{tabular}{cc}
\includegraphics[width=3in]{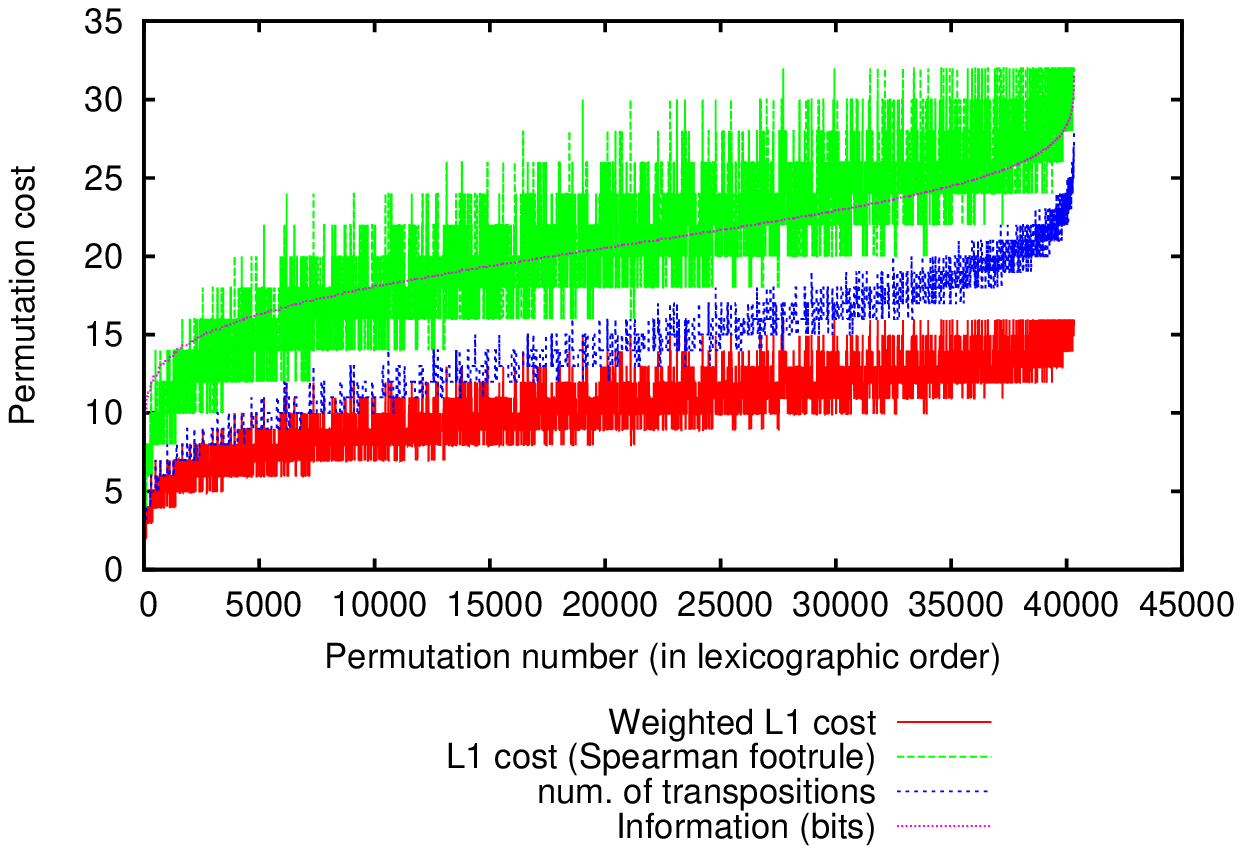}&
\includegraphics[width=3in]{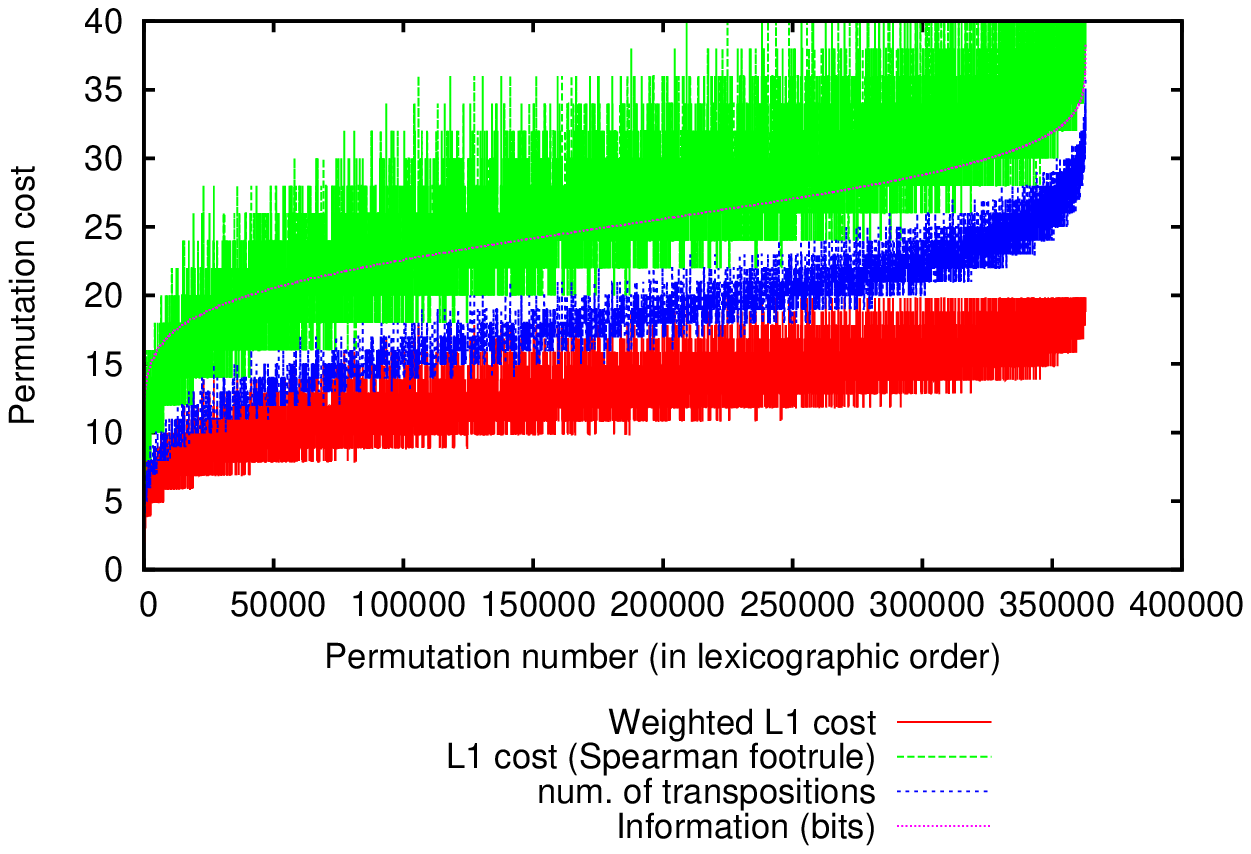}\\
(a) & (b) \\
\end{tabular}
\caption{Variation of costs over the set of all permutations in the the symmetric groups (a) $S_8$, and (b) $S_9$. Spearman's foot rule distance is in Green. Kendall tau distance is in Blue. Canberra distance is given in Red. Information measure defined in this work is given in Magenta.}
\label{fig:permcost}
\end{figure}

\section{Results}
\label{sec:results} 
We first quantify the effect of disarray between permutations as assessed by various popular measures.  
Figure \ref{fig:permcost} gives the cost associated with various measures for the set of all permutations in symmetric groups (a) $S_7$ and (b) $S_8$. Specifically, the measures used are: (1) Spearman's foot rule metric ($L_1$ distance), (2) Canberra distance (weighted $L_1$ measure), (3) Kendall's tau distance, measuring the number of adjacent transpositions to sort a permutation, and (4) the information measure we developed in this work. It can be seen that as the information content to describe a permutation increases, all the other measures vary significantly. It is important to note that  the costs reported by all four of the considered measures are related to the total number of adjacent transpositions of elements required to sort the permutation. However, our measure of information accounts for the varying magnitude of disarray (given by the permutation's factoradic digits) of each element, instead of combining and summarizing using a simple number. Other measures overlook these individual contributions; for instance, it can be seen from Figure \ref{fig:factoradic} that the permutations $adcb =5_{10} = (0,2,1,0)_!$, $bcda =9_{10} = (1,1,1,0)_!$, $bdca =10_{10} = (1,2,0,0)_!$, $cadb =13_{10}=(2,0,1,0)_!$, and $dabc =18_{10}=(3,0,0,0)_!)$ all require the same number of transpositions ($=3$), yet differing in the number of individual transpositions required by its elements.  

To examine the performance of various measures on comparing \topklists, we first consider three top 250 movie lists downloaded from \url{goodmovieslist.com}, \url{imdb.com} and \url{reddit.com}. Figure \ref{fig:movies} shows the comparisons of (left to right)  Goodmovies vs. IMDb, Goodmovies vs. Reddit, and IMDb vs. Reddit, while varying $k$ from 1 to 250 in increments of one.

Qualitatively the lists corresponding to Goodmovies and IMDb  are more similar than the other possible pairs.  It can be seen from the figure that both Spearman's foot rule distance and Kentall tau distance grow roughly quadratically with the size of $k$. This mainly results from the contributions to the respective costs from the set of non-overlapping elements. As this set grows, its contribution to the distance dominates. However, the growth of information cost\footnote{We note that, for these results and those to follow, we compute the measure of information between lists by assuming that the total domain of movies is unknown, \ie\ following case 2 described in Section \ref{sec:case2}.}
is roughly linear. This makes more sense, as information is additive. When the size of the list increases from $k$ to $k+1$, the new element that gets added to each of the two lists can in the worst case be independent of the previous information. This implies that, in the worst case, the total information content in the list going from $k$ to $k+1$ gets augmented by the sum of information in the new elements. Therefore, the quadratic grown of the other measures is questionable.

It is interesting to note that while Spearman's foot rule and Kendall tau distances monotonically increase, the information cost plotted in the figure has `fluctuations' in the amount of information measured. These variations occurs when new elements (for increasing values of $k$)  cause the set of overlapping elements to grow in size. While this increases the cost to state the permutation of overlapping elements, there is a net saving because the size of the set of non-overlapping elements (which are transmitted using $LZW$ compression) decreases, in comparison with the previous values of $k$.

\begin{figure}
\hspace{-0.5in}
\includegraphics[width=1.6in, angle=-90]{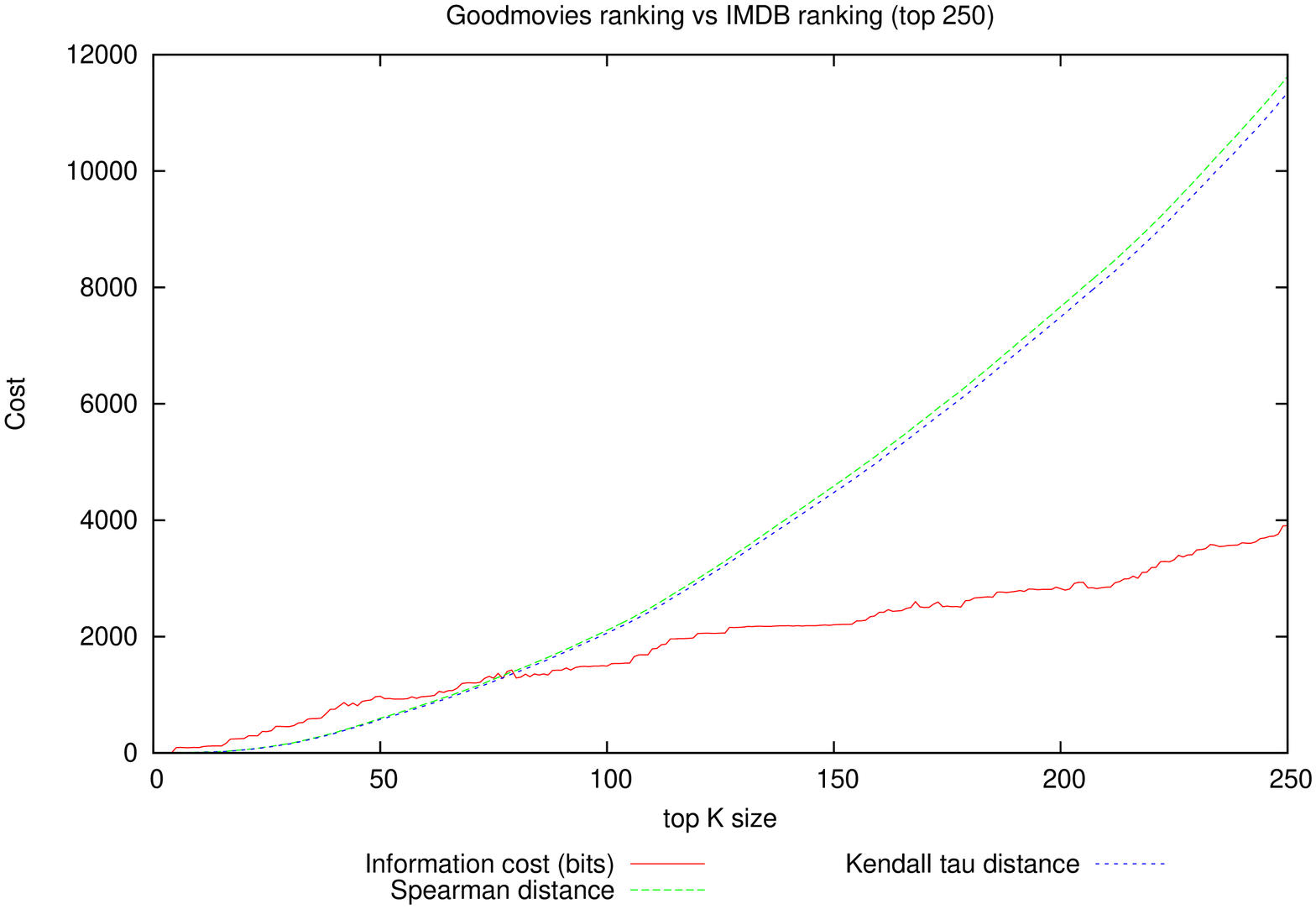}
\includegraphics[width=1.6in, angle=-90]{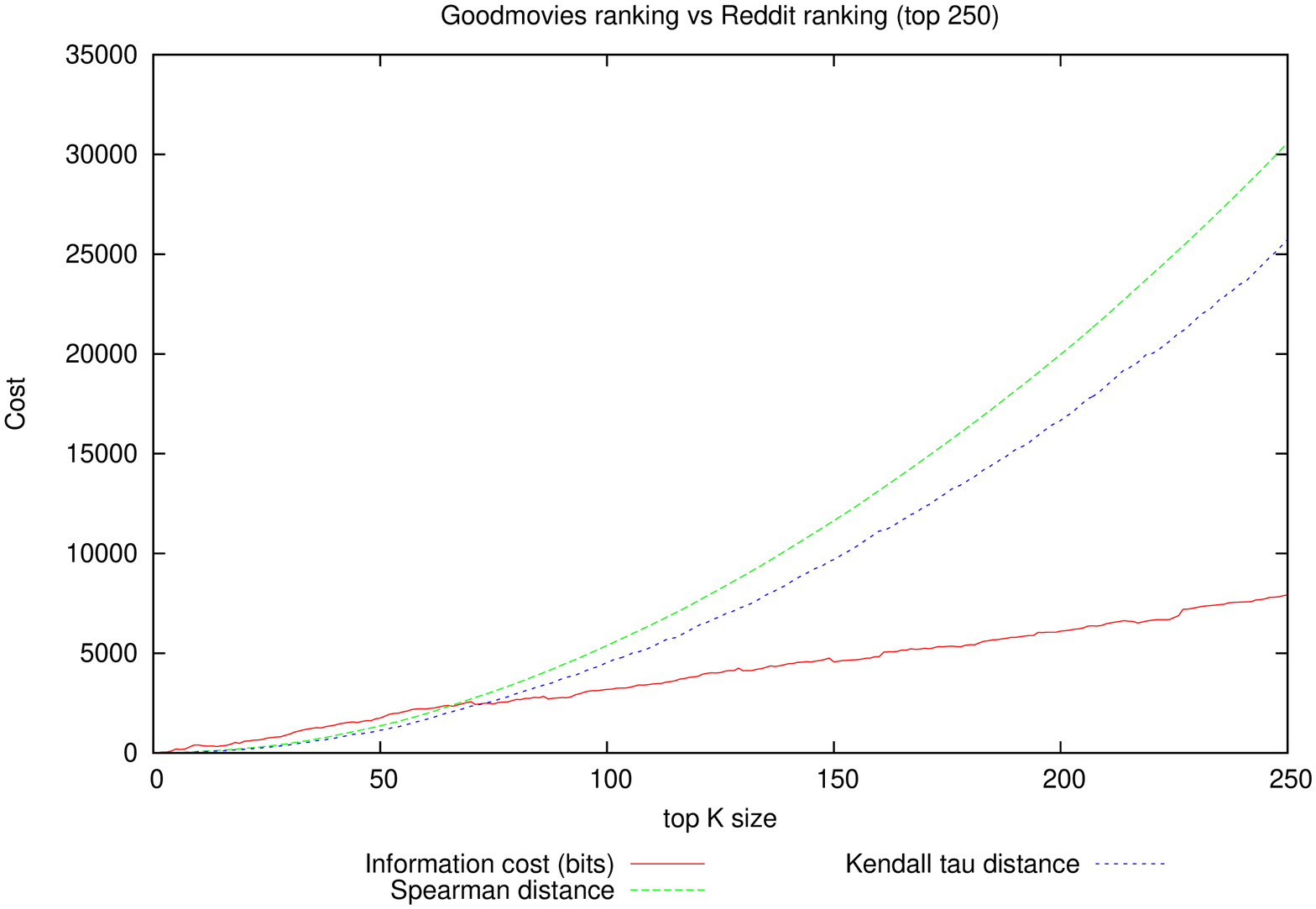}
\includegraphics[width=1.6in, angle=-90]{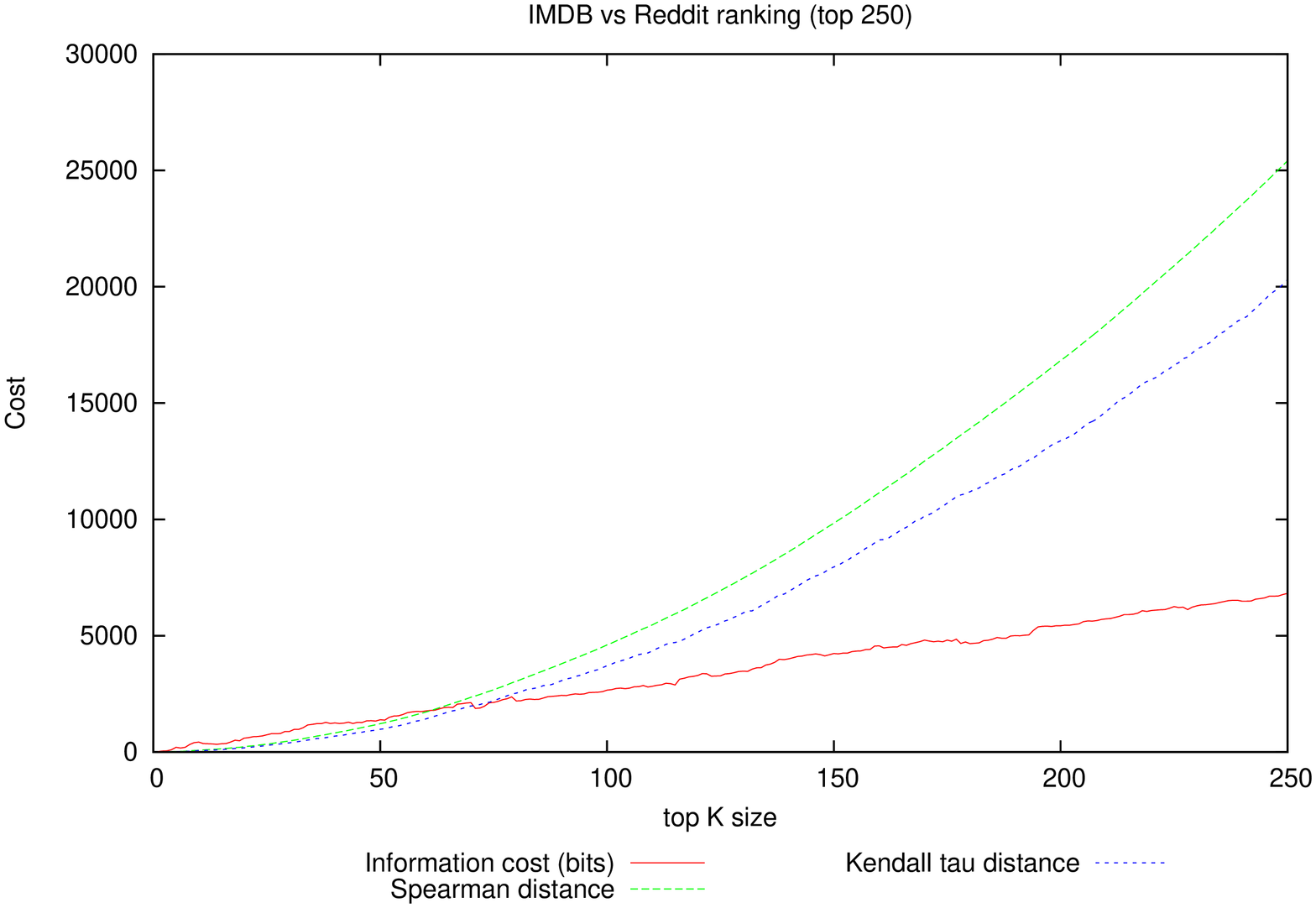}
\caption{Comparison of Spearman's foot rule distance, Kendall tau and Information distance on the rankings from \texttt{goodmovieslist.com}, \texttt{imdb.com} and \texttt{reddit.com}}
\label{fig:movies}
\end{figure}

To undertake this comparison in a large scale, we compare the search results of three popular web search engines: Google, Yahoo and Ask. We do this by selecting  250 top trending search and news terms reported by Google Trends and Yahoo text Analytics for the regions of Australia, US, India, Canada, UK, Singapore and Germany.  Figure \ref{fig:search} plots the average (mean) cost over all the 250 queries computing using Information, Spearman's foot rule and Kendall tau measures. In this experiment, we vary $k$ as $10,25,50,75,$ and $100$. In this figure the same growth trends witness previously emerges, linear for information cost and quadratic for Spearman and Kendall distance. 
For $k>50$, the difference between the avg costs between each pair of search engine results grows drastically for Spearman and Kendall distances, while the same using information cost does not.

\begin{figure}
\hspace{-0.5in}
\includegraphics[width=1.6in, angle=-90]{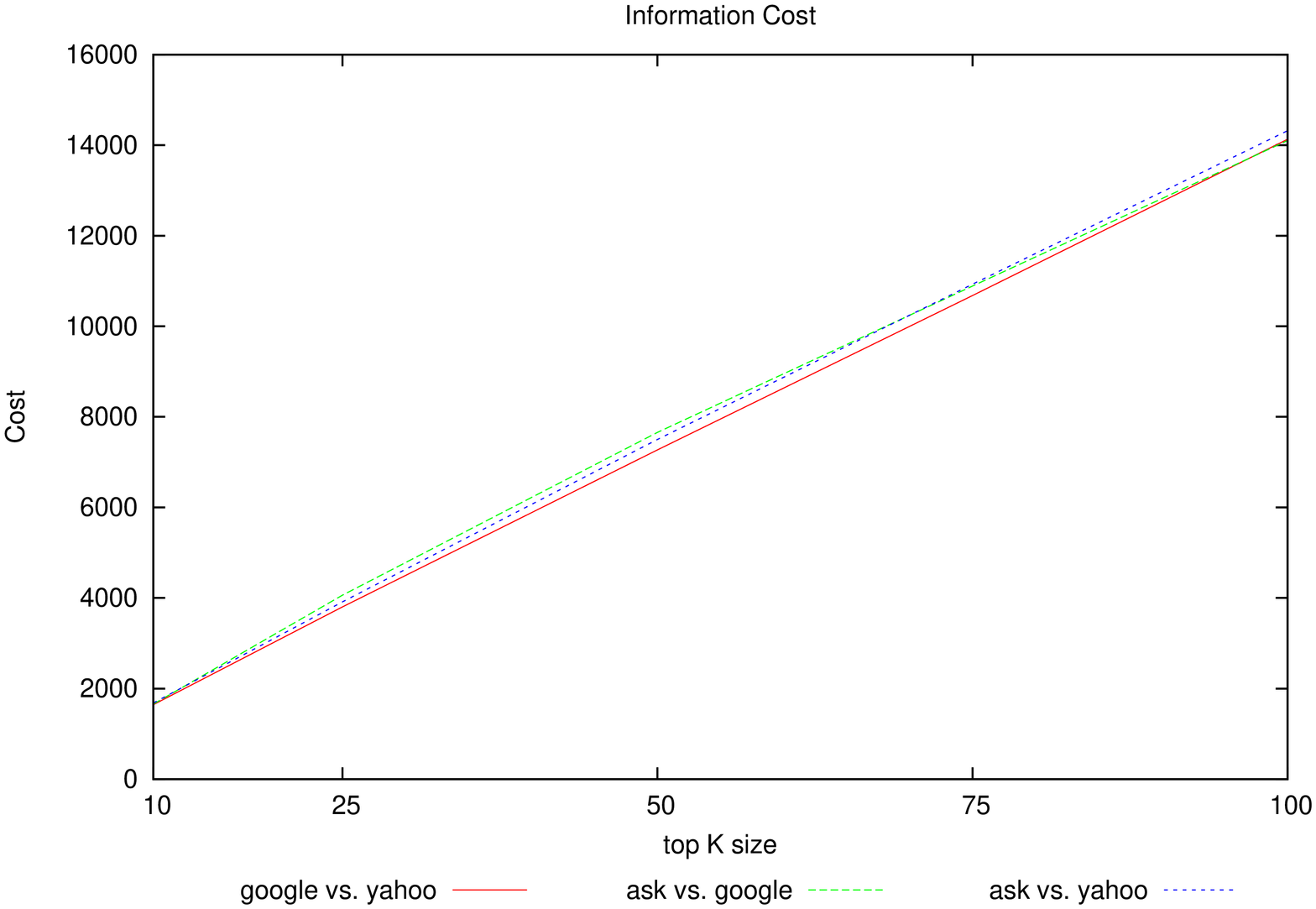}
\includegraphics[width=1.6in, angle=-90]{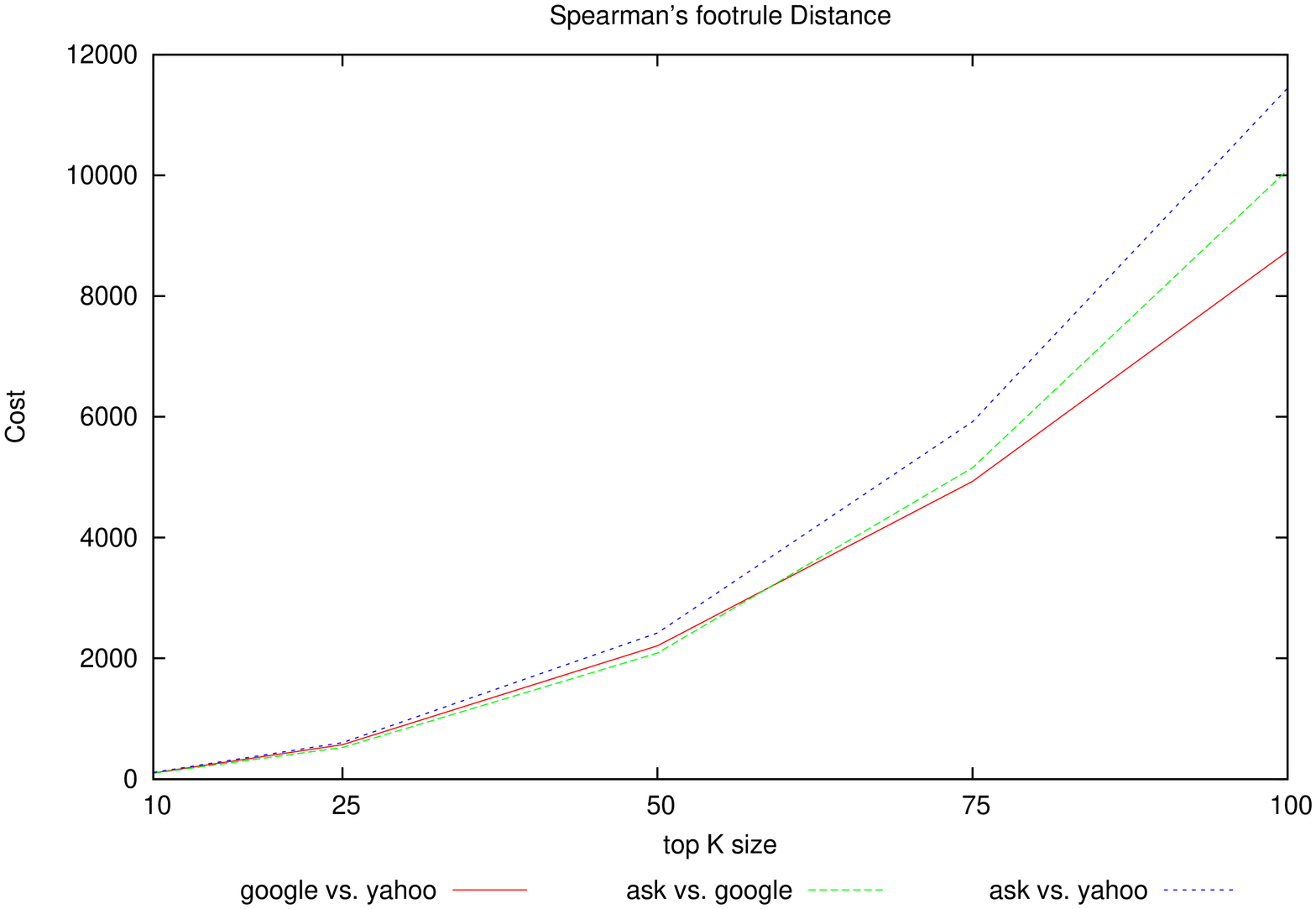}
\includegraphics[width=1.6in, angle=-90]{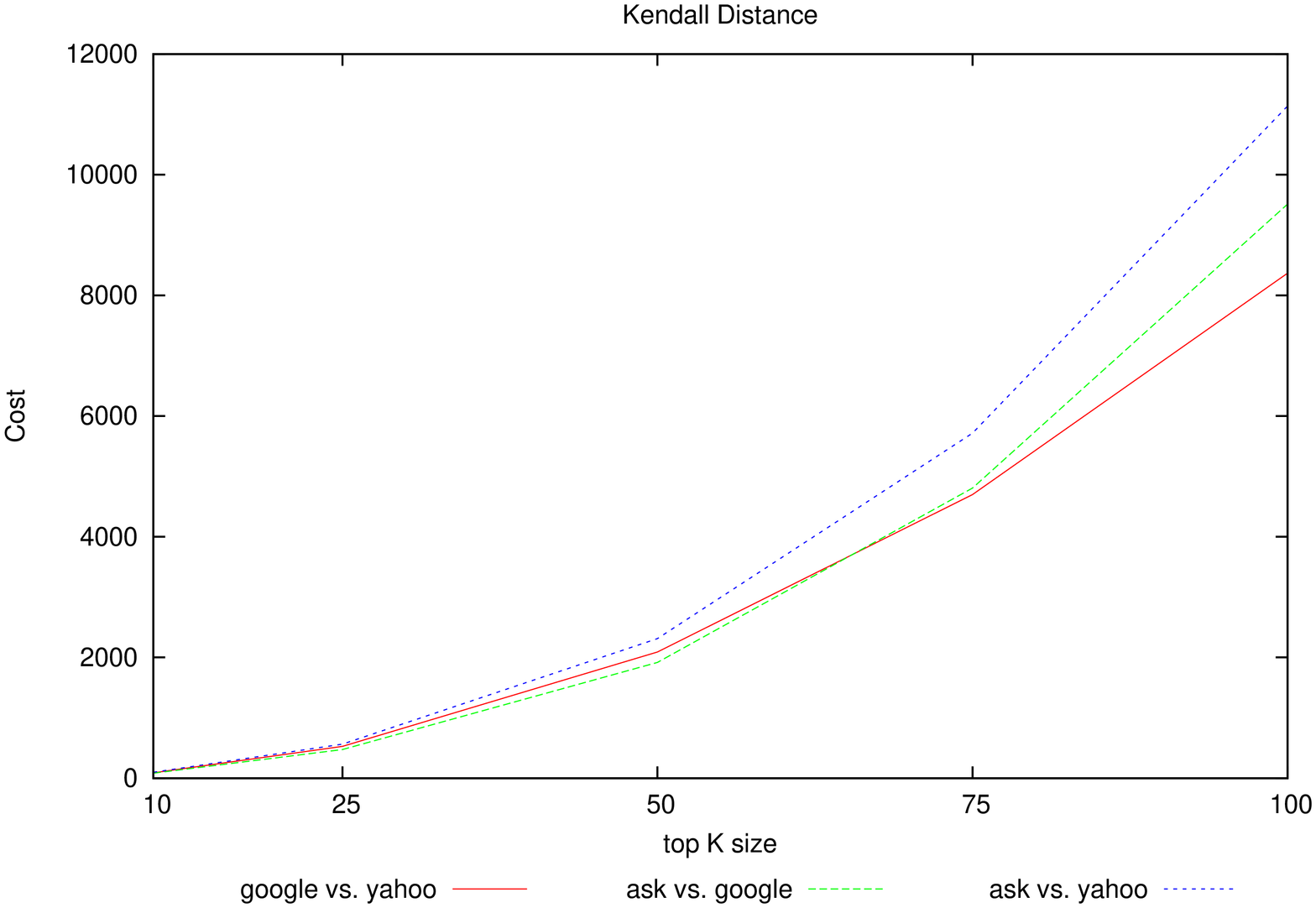}

\caption{Comparison of Information distance, Spearman's foot rule distance, and Kendall tau on search results return by Google, Yahoo and Ask. The reported values are averaged over 250 search terms comparing pairs of ranked lists for values of $k= \{10,25,75, 100\}$}
\label{fig:search}
\end{figure}

\section{Conclusion}
We have introduced a new information measure for comparing any two top $k$ lists. By exploring their compressibility, our method provides a statistically rigorous measure of variability between ranked lists. It provides an objective trade-off between criteria that measure the dis-similarity between lists, addressing the lacunae and pitfalls in the existing measures. As a future direction of research, this measure can be used to address the important \textit{rank aggregation problem}: What is the `consensus' top $k$ ranking that combines the top $k$ results from multiple sources.

~\\
\noindent\textbf{Acknowledgments.} We thank Chetana Gavankar for rekindling our interest on this problem. ASK thanks Lloyd Allison for numerous discussions and helpful suggestions on this topic.

\bibliographystyle{acm}
\bibliography{soda}
\end{document}